\documentclass[11pt]{article}

\makeatletter
\renewcommand*\@fnsymbol[1]{\the#1}
\makeatother

\usepackage{amsmath}
\usepackage{amsfonts}
\usepackage{amsthm}
\usepackage[english]{babel}
\usepackage{latexsym}
\usepackage[hmargin=2cm,vmargin=4cm]{geometry}
\usepackage{enumerate}
\usepackage{nicefrac}
\usepackage{mathrsfs}
\usepackage{comment}
\usepackage[affil-it]{authblk}

\theoremstyle{plain}
\newtheorem{theorem}{Theorem}[section]
\newtheorem{lemma}[theorem]{Lemma}
\newtheorem{proposition}[theorem]{Proposition}
\newtheorem{corollary}[theorem]{Corollary}
\theoremstyle{definition}
\newtheorem{definition}[theorem]{Definition}
\theoremstyle{remark}
\newtheorem{remark}[theorem]{Remark}
\newtheorem{example}[theorem]{Example}

\newcommand{\Indexfin}{\mathop{\rm fin}\nolimits}
\newcommand{\Indexiqr}{\mathop{\rm iqr}\nolimits}

\newcommand{\Closure}{\mathop{\rm cl}\nolimits}

\newcommand{\e}{\varepsilon}
\newcommand{\VaR}{\mathop{\rm VaR}\nolimits}
\newcommand{\TVaR}{\mathop{\rm TVaR}\nolimits}

\newcommand{\cF}{{\mathcal F}}
\newcommand{\cM}{{\mathcal{M}}}
\newcommand{\cN}{{\mathcal N}}

\newcommand{\cR}{{\mathcal R}}

\newcommand{\E}{{\mathbb E}}
\newcommand{\N}{{\mathbb N}}
\newcommand{\probp}{{\mathbb P}}
\newcommand{\probq}{{\mathbb Q}}
\newcommand{\R}{{\mathbb R}}

\newcommand{\cA}{\mathscr A}

\newcommand{\cL}{\mathscr L}
\newcommand{\cS}{{\mathscr{S}}}

\newcommand{\cD}{{\mathscr{D}}}
\newcommand{\cU}{{\mathscr{U}}}

\newcommand{\cX}{{\mathscr{X}}}

\def\keywords{\vspace{.5em}
{\noindent\textbf{Keywords}:\,\relax%
}}

\def\MSCclassification{\vspace{.5em}
{\noindent\textbf{MSC}:\,\relax%
}}

\makeatletter
\def\@fnsymbol#1{\ensuremath{\ifcase#1\or *\or 1\or 2\or
   3\or 4\or 5\or 6\or 7\or 8\else\@ctrerr\fi}}
\makeatother

\begin{document}

\title{Law-invariant risk measures:\\
extension properties and qualitative robustness\footnote{Partial support through the SNF project 51NF40-144611 ``Capital adequacy, valuation, and portfolio selection for insurance companies'' is gratefully acknowledged.}}

\author{\sc{Pablo Koch-Medina}\thanks{Email: \texttt{pablo.koch@bf.uzh.ch}}}
\affil{Department of Banking and Finance, University of Zurich, Switzerland}

\author{\sc{Cosimo Munari}\,\thanks{Email: \texttt{cosimo.munari@math.ethz.ch}}}
\affil{Department of Mathematics, ETH Zurich, Switzerland}

\date{January 9, 2014}

\maketitle

\begin{abstract}
We characterize when a convex risk measure associated to a law-invariant acceptance set in $L^\infty$ can be extended to $L^p$, $1\leq p<\infty$, preserving \textit{finiteness} and \textit{continuity}. This problem is strongly connected to the statistical robustness of the corresponding risk measures. Special attention is paid to concrete examples including risk measures based on expected utility, max-correlation risk measures, and distortion risk measures.
\end{abstract}


\MSCclassification{91B30, 91G80}


\parindent 0em \noindent

\section{Introduction}

The objective of this paper is to complement the paper~\cite{FilipovicSvindland2012} by Filipovi\'{c} and Svindland. The main result in that paper is Theorem 2.2 stating that every convex, law-invariant, lower semicontinuous map $f:L^\infty\to\R\cup\{\infty\}$ can be uniquely extended to a map on $L^1$ satisfying the same properties. In this sense, $L^1$ can be viewed as the canonical space for this type of maps. The results in~\cite{FilipovicSvindland2012} are presented in the context of a standard probability space but can be extended to a nonatomic setting as shown in Svindland~\cite{Svindland2010}.

\smallskip

The authors in \cite{FilipovicSvindland2012} are mostly concerned with the application of their results in the context of cash-additive risk measures. It is well-known that any cash-additive risk measure on $L^\infty$ is automatically finite-valued and (Lipschitz) continuous. It is also well-known that cash-additive risk measures on $L^p$, $1\leq p<\infty$, need not be either finite-valued or continuous. Consequently, the following two questions arise in a natural way: \textit{When does a cash-additive risk measure on $L^\infty$ admit a finite-valued, continuous extension to $L^p$ for a given $1\leq p<\infty$, and, which is the ``largest'' space $L^p$ for which such an extension exists?}

\smallskip

We provide a full answer to the previous questions, characterizing those convex, law-invariant, cash-additive risk measures defined on $L^\infty$ that can be extended to $L^p$ spaces preserving finiteness and continuity. In fact, our main result provides a characterization in the more general setting of risk measures studied by Farkas, Koch-Medina, and Munari in \cite{FarkasKochMunari2013a}, where cash additivity is not required. More precisely, we show that the existence of finite, continuous extensions depends on the properties of the underlying \textit{acceptance sets}. Special attention is paid to several concrete examples, including risk measures based on expected utility, max-correlation risk measures, and distortion risk measures. These examples show that, if finiteness and continuity are to be preserved, the ``canonical'' model space for convex, law-invariant risk measures is not always $L^1$ but can be any space $L^p$, $1\le p\le\infty$. In particular, there are (even cash-additive) risk measures that cannot be extended beyond $L^\infty$ maintaining finiteness and continuity.

\smallskip

The previous questions turn out to be intimately related to the \textit{statistical robustness} for risk measures as discussed by Kr\"{a}tschmer, Schied, and Z\"{a}hle in \cite{KraetschmerSchiedZaehle2013}, highlighting the practical relevance of our results and our examples. Indeed, when risk measures are implemented to define capital adequacy requirements for financial institutions, or margin requirements for the participants of a central exchange, their statistical robustness is crucial to guarantee that the corresponding capital requirements are stable with respect to small changes in the distributions of the underlying positions. In this respect, our examples can be seen to be complementary to \cite{KraetschmerSchiedZaehle2013}.


\section{Preliminaries}

Let $\cX$ be an ordered topological vector space over $\R$ with positive cone $\cX_+$ and topological dual $\cX'$. The space $\cX$ represents the set of all possible \textit{capital positions} -- assets net of liabilities -- of financial institutions at a fixed future date $t=T$.

\medskip

We assume $\cA\subset\cX$ is an \textit{acceptance set}, i.e., a nonempty, proper subset of $\cX$ satisfying $\cA+\cX_+\subset\cA$. We interpret the elements of $\cA$  as those capital positions which are deemed acceptable by an external or internal ``regulator''. Moreover, let $S=(S_0,S_T)$ represent a \textit{traded asset} with price $S_0>0$ at time $t=0$ and nonzero, terminal payoff $S_T\in\cX_+$ at time $t=T$. The \textit{risk measure} associated to $\cA$ and $S$ is the map $\rho_{\cA,S}:\cX\to\overline{\R}$ defined by
\begin{equation}
\label{risk measure intro}
\rho_{\cA,S}(X):=\inf\left\{m\in\R \,; \ X+\frac{m}{S_0}S_T\in\cA\right\}\,.
\end{equation}
For a position $X\in\cX$, the quantity $\rho_{\cA,S}(X)$ represents the ``minimum'' amount of capital that needs to be raised and invested in the asset $S$ to guarantee acceptability. Clearly, a negative $\rho_{\cA,S}(X)$ implies that capital is returned to shareholders. The motivation for studying this type of risk measures is discussed in detail in \cite{FarkasKochMunari2013a}, where general results on finiteness and continuity are also provided.

\medskip

We are particularly interested in the case where $\cX$ is the Banach space $L^p$, $1\leq p\leq\infty$, defined over a probability space $(\Omega,\cF,\probp)$, which we assume to be \textit{nonatomic}. The corresponding norm is denoted by $\left\|\cdot\right\|_p$. The space $L^p$ becomes a Banach lattice when additionally equipped with the canonical order structure, i.e., $X\ge Y$ whenever this inequality holds almost surely. The \textit{conjugate} of $p$ will be denoted by $p'$, i.e. we set $p':=\frac{p}{1-p}$.

\medskip

When $\cX=L^p$ for some $1\leq p\leq\infty$, we can consider risk measures $\rho_{\cA,S}$ with respect to the {\em cash asset} $S=(1,1_\Omega)$. These risk measures are called \textit{cash additive}. We refer to \cite{FoellmerSchied2011} for a comprehensive treatment when $p=\infty$. For cash-additive risk measures, we simply write for $X\in L^p$
\begin{equation}
\rho_\cA(X) := \rho_{\cA,S}(X) = \inf\{m\in\R \,; \ X+m\in\cA\}\,.
\end{equation}
Note that our general setting also allows for a traded asset $S=(S_0,S_T)$ with an arbitrary, positive, random payoff. Hence, we can also cover situations were no risk-free security exists, as discussed in \cite{FarkasKochMunari2013a}.


\section{The general extension theorem}

In this section we provide the key extension result for risk measures of the form $\rho_{\cA,S}$. Although our main interest lies in convex risk measures on $L^p$ spaces, to highlight the underlying structure of our result we first study extension theorems in the setting of abstract ordered topological vector spaces. Throughout this section $\cL$ and $\cS$ will denote two ordered topological vector spaces over $\R$ with respective positive cones $\cL_+$ and $\cS_+$. We assume that $\cS$ is a dense subspace of $\cL$. Hence, in addition to its own topology, the space $\cS$ can be equipped with the relative topology induced by $\cL$, which we call the $\cL$-{\em topology}. Moreover, since every functional on $\cL$ can be restricted to a functional on $\cS$, we may also consider the weak topology $\sigma(\cS,\cL')$ on $\cS$ where, by abuse of notation, we do not distinguish between functionals on $\cL$ and their restrictions to $\cS$. In the next section we will take $\cL=L^p$, for some $1\leq p<\infty$, and $\cS=L^\infty$.

\medskip

The following theorem is our main result. For a subset $\cA\subset\cL$, we denote by $\Closure_\cL(\cA)$ the closure of $\cA$ with respect to the topology on $\cL$.

\begin{theorem}
\label{main theorem}
Let $\cA\subset\cS$ be a convex, $\sigma(\cS,\cL')$-closed acceptance set and $S=(S_0,S_T)$ a traded asset with $S_T\in\cS_+$. Assume $\rho_{\cA,S}$ is finite-valued and continuous on $\cS$. The following statements are equivalent:
\begin{enumerate}[(a)]
	\item $\rho_{\cA,S}$ can be extended to a finite-valued, continuous risk measure on $\cL$;
	\item $\rho_{\cA,S}$ is continuous at $0$ with respect to the $\cL$-topology;
	\item $\cA$ has nonempty interior with respect to the $\cL$-topology;
	\item $\Closure_\cL(\cA)$ has nonempty interior in $\cL$.
\end{enumerate}
In this case, the extension is unique and is given by $\rho_{\Closure_\cL(\cA),S}$.
\end{theorem}

\medskip

Before proving Theorem~\ref{main theorem}, it is useful to collect some auxiliary results.


\subsubsection*{Dual representations and continuity}

Let $\cX$ be an ordered topological vector space over $\R$ with positive cone $\cX_+$ and topological dual $\cX'$. By $\cX'_+$ we denote the set of all functionals $\psi\in\cX'$ satisfying $\psi(X)\geq0$ for every $X\in\cX_+$. We start by stating a dual representation result for convex risk measures of the form $\rho_{\cA,S}$ in a version that is convenient for our purposes.

\medskip

The (lower) \textit{support function} of a subset $\cA\subset\cX$ is the map $\sigma_\cA:\cX'\to\R\cup\{-\infty\}$ defined by
\begin{equation}
\sigma_\cA(\psi):=\inf_{A\in\cA}\psi(A)\,.
\end{equation}
The set
\begin{equation}
B(\cA):=\{\psi\in\cX' \,; \ \sigma_\cA(\psi)>-\infty\}
\end{equation}
is called the {\em barrier cone} of $\cA$. Clearly, the support function of a set $\cA\subset\cX$ always coincides with the support function of its closure.

\medskip

By combining Corollary 4.14 and Theorem 4.16 in \cite{FarkasKochMunari2013b} we obtain the following result.

\begin{lemma}
\label{lemma dual repr}
Let $\cA\subset\cX$ be a closed, convex acceptance set and $S=(S_0,S_T)$ a traded asset, and set
\begin{equation}
\cX'_{+,S}:=\{\psi\in\cX'_+ \,; \ \psi(S_T)=S_0\}\,.
\end{equation}
Then the following statements are equivalent:
\begin{enumerate}[(a)]
	\item $\rho_{\cA,S}$ attains some finite value;
	\item $\rho_{\cA,S}$ does not attain the value $-\infty$;
	\item $\cX'_{+,S}\cap B(\cA)$ is nonempty.
\end{enumerate}
In this case, for every $X\in\cX$ we have
\begin{equation}
\label{dual representation auxiliary 1}
\rho_{\cA,S}(X)=\sup_{\psi\in\cX'_{+,S}}\{\sigma_\cA(\psi)-\psi(X)\}\,.
\end{equation}
\end{lemma}

\smallskip

\begin{remark}
\label{remark on application of lemma}
Later we will apply this result to the situation where $\cX$ is either $\cL$, equipped with its own topology, or $\cS$, equipped with the $\sigma(\cS,\cL')$ topology. Since both of these spaces have the same dual $\cL^\prime$ we see that for a subset $\cA\subset\cS$
\begin{equation}
\sigma_\cA(\psi) = \sigma_{\Closure_\cL(\cA)}(\psi) \ \ \ \mbox{for all} \ \psi\in\cL'\,,
\end{equation}
where $\sigma_\cA$ is applied to the restriction of $\psi\in\cL'$ to $\cS$. In particular, the intersection $\cL'_{+,S}\cap B(\cA)$ is nonempty if and only if $\cL'_{+,S}\cap B(\Closure_\cL(\cS))$ is also nonempty.
\end{remark}

\medskip

We now recall a necessary and a sufficient condition for a risk measure of the form $\rho_{\cA,S}$ to be continuous on $\cX$. For a proof, we refer to Lemma 2.5 and Theorem 3.10 in \cite{FarkasKochMunari2013a}, respectively.

\begin{lemma}
\label{lemma transversality}
Let $\cA\subset\cX$ be an acceptance set and $S=(S_0,S_T)$ a traded asset.
\begin{enumerate}[(i)]
	\item If $\rho_{\cA,S}$ is continuous at some point $X\in\cX$ with $\rho_{\cA,S}(X)<\infty$, then $\cA$ has nonempty interior.
	\item Assume $\cA$ is convex and has nonempty interior. Then $\rho_{\cA,S}$ is continuous whenever it is finite-valued.
\end{enumerate}
\end{lemma}


\subsubsection*{Proof of Theorem~\ref{main theorem}}

Assume \textit{(a)} holds, i.e., $\rho_{\cA,S}$ admits an extension to a finite-valued, continuous map on $\cL$. Then, clearly, $\rho_{\cA,S}$ must also be continuous with respect to the $\cL$-topology. It follows that \textit{(b)} holds.

\smallskip

Let \textit{(b)} hold so that $\rho_{\cA,S}$ is continuous at $0$ with respect to the $\cL$-topology. Since $\rho_{\cA,S}$ is finite at $0$, Lemma \ref{lemma transversality} implies that $\cA$ must have nonempty interior with respect to the $\cL$-topology, hence \textit{(c)} holds.

\smallskip

Now, assume that \textit{(c)} holds so that $\cA$ has nonempty interior with respect to the $\cL$-topology. As a result, we find an open subset $\cU$ of $\cL$ such that $\cU\cap\cS\subset\cA$. Since $\cS$ is dense in $\cL$ and $\cU$ is open in $\cL$, we have that $\cU\cap\cS$ is nonempty and
\begin{equation}
\Closure_\cL(\cU) = \Closure_\cL(\cU\cap\cS)\,.
\end{equation}
It follows that
\begin{equation}
\cU \subset \Closure_\cL(\cU) = \Closure_\cL(\cU\cap\cS) \subset \Closure_\cL(\cA)\,,
\end{equation}
proving that $\Closure_\cL(\cA)$ has nonempty interior in $\cL$ and, hence, that \textit{(d)} holds.

\smallskip

Finally, assume that \textit{(d)} holds, i.e., $\Closure_\cL(\cA)$ has nonempty interior in $\cL$. Since $\cA$ is convex and $\sigma(\cS,\cL')$-closed and $\rho_{\cA,S}$ is finite-valued on $\cS$, Lemma~\ref{lemma dual repr} implies that $\cL'_{+,S}\cap B(\cA)$ is nonempty and
\begin{equation}
\rho_{\cA,S}(X)=\sup_{\psi\in\cL'_{+,S}}\{\sigma_\cA(\psi)-\psi(X)\}
\end{equation}
for every $X\in\cS$.

\smallskip

In particular, it follows that $\cL'_{+,S}\cap B(\Closure_\cL(\cA))$ is nonempty by Remark \ref{remark on application of lemma}, so that we can apply Lemma~\ref{lemma dual repr} once again to obtain
\begin{equation}
\rho_{\Closure_\cL(\cA),S}(X) = \sup_{\psi\in\cL'_{+,S}}\{\sigma_{\Closure_\cL(\cA)}(\psi)-\psi(X)\} = \sup_{\psi\in\cL'_{+,S}}\{\sigma_\cA(\psi)-\psi(X)\}
\end{equation}
for all $X\in\cL$. This shows that $\rho_{\Closure_\cL(\cA),S}$ extends $\rho_{\cA,S}$ to the whole of $\cL$. We claim that $\rho_{\Closure_\cL(\cA),S}$ is finite-valued and continuous.

\smallskip

Indeed, note first that, again by Lemma \ref{lemma dual repr}, the map $\rho_{\Closure_\cL(\cA),S}$ cannot take the value $-\infty$. Consider now the convex set
\begin{equation}
\cD := \{X\in\cL \,; \ \rho_{\Closure_\cL(\cA),S}(X)<\infty\} = \{X\in\cL \,; \ \rho_{\Closure_\cL(\cA),S}(X)\in\R\}\,.
\end{equation}
Since $\Closure_\cL(\cA)\subset\cD$, the interior of $\cD$ is nonempty. Now, assume there exists $X\in\cL\setminus\cD$. In this case, by a standard separation argument we find a nonzero $\psi\in\cL'$ such that $\psi(X)\leq\psi(Y)$ for every $Y\in \cD$. Since $\rho_{\cA,S}$ is finite-valued, the set $\cD$ contains the subspace $\cS$. As a consequence, $\psi$ must annihilate $\cS$ and therefore, by density, the whole of $\cL$. This is not possible since $\psi$ was nonzero. Hence, $\cD=\cL$ showing that $\rho_{\Closure_\cL(\cA),S}$ is finite-valued and, by Lemma \ref{lemma transversality}, continuous on $\cL$. It follows that \textit{(d)} implies \textit{(a)}.

\smallskip

We conclude the proof of Theorem~\ref{main theorem} by observing that any continuous extension of $\rho_{\cA,S}$ must be unique because $\cS$ is dense in $\cL$.


\section{Extension of risk measures on $L^p$ spaces}

We now apply Theorem~\ref{main theorem} to risk measures of the form $\rho_{\cA,S}$ on $L^p$ spaces when the underlying acceptance set $\cA$ is convex and law-invariant. Recall that a set $\cA\subset L^p$ is called \textit{law-invariant} if $X\in\cA$ whenever $X\sim Y$ for some $Y\in\cA$. Here, we write $X\sim Y$ to indicate that $X$ and $Y$ have the same law. Similarly, a map $f:L^p\to\overline{\R}$ is said to be \textit{law-invariant} if $f(X)=f(Y)$ whenever $X\sim Y$.

\medskip

\begin{remark}
\label{remark on law invariant sets}
Recall that every closed, convex, law-invariant set $\cA\subset L^\infty$ is also closed with respect to the $\sigma(L^\infty,L^p)$-topology for any $1\leq p\leq\infty$. This property was first shown in the context of standard probability spaces by Jouini, Schachermayer and Touzi, see Remark 4.4 in \cite{JouiniSchachermayerTouzi}, and was later extended to general nonatomic spaces by Svindland, see Proposition 1.2 in \cite{Svindland2010}.
\end{remark}

\smallskip

\begin{remark}
\label{lack law invariance risk measure}
Consider a law-invariant acceptance set $\cA\subset L^\infty$ and a traded asset $S=(S_0,S_T)$. It is immediate to see that the risk measure $\rho_{\cA,S}$ is also law-invariant whenever the payoff $S_T$ is deterministic. In particular, the cash-additive risk measure $\rho_\cA$ is always law-invariant if $\cA$ is law-invariant. However, this need not be the case if $S_T$ is genuinely random, as illustrated by the following important example.

\begin{enumerate}
	\item The \textit{Value-at-Risk} and \textit{Tail Value-at-Risk} of $X\in L^\infty$ at the level $0<\alpha<1$ are defined, respectively, by
\begin{equation}
\VaR_\alpha(X) := \inf\{m\in\R \,; \ \probp(X+m<0)\leq\alpha\}\,,
\end{equation}
\begin{equation}
\TVaR_\alpha(X) := \frac{1}{\alpha}\int^\alpha_0\VaR_\beta(X)d\beta\,.
\end{equation}
As is well-known, e.g. \cite{FoellmerSchied2011}, both are law-invariant, cash-additive risk measures. Moreover, $\TVaR_\alpha$ is always coherent while $\VaR_\alpha$ is not convex in general. In particular, we can consider the law-invariant acceptance set based on Tail Value-at-Risk at level $\alpha$
\begin{equation}
\cA^\alpha := \left\{X\in L^\infty \,; \ \TVaR_\alpha(X)\leq0\right\}\,.
\end{equation}
We claim that $\rho_{\cA^\alpha,S}$ is never law-invariant unless $S_T$ is deterministic, in which case $\rho_{\cA^\alpha,S}$ is just a multiple of $\TVaR_\alpha$.

\smallskip

To see this, assume $S_T$ is not deterministic so that there exist $\gamma_2>\gamma_1>0$ for which $\probp(S_T\leq\gamma_1)>0$ and $\probp(S_T\geq\gamma_2)>0$. Since $(\Omega,\cF,\probp)$ is nonatomic, we can find measurable sets $A\subset\{S_T\leq\gamma_1\}$ and $B\subset\{S_T\geq\gamma_2\}$ satisfying $\probp(A)=\probp(B)=p$ with $0<p<1-\alpha$. Set now $C=(A\cup B)^c$ and note that $\probp(C)=1-2p$. For $-\gamma_2<\lambda<-\gamma_1$ we define
\begin{equation}
X := \lambda1_A-S_T1_C \ \ \ \mbox{and} \ \ \ Y := \lambda1_B-S_T1_C\,.
\end{equation}
Then clearly $X\sim Y$. We now show that $\rho_{\cA^\alpha,S}(X)>S_0\geq\rho_{\cA^\alpha,S}(Y)$, implying that $\rho_{\cA^\alpha,S}$ is not law-invariant.

\smallskip

Indeed, note first that $Y+S_T\geq0$ so that $\rho_{\cA^\alpha,S}(Y)\leq S_0$. Now take $m<0$. Then
\begin{equation}
\probp(X+S_T+m<0)\geq\probp(A)+\probp(C)=1-p>\alpha\,,
\end{equation}
implying $\VaR_\beta(X+S_T)\geq0$ for all $0<\beta\leq\alpha$ and, hence, $\TVaR_\alpha(X+S_T)\geq0$. Moreover, if $0\leq m<-\lambda-\gamma_1$ then $\probp(X+S_T+m<0)=\probp(A)=p$, and therefore $\VaR_\beta(X+S_T)\geq-\lambda-\gamma_1>0$ whenever $0<\beta<p$. It follows that $\TVaR_\alpha(X+S_T)>0$, showing that $\rho_{\cA^\alpha,S}(X)>S_0$ as claimed.

  \item Sometimes $\rho_{\cA,S}$ is law-invariant even if $S_T$ is not deterministic. For example, consider the law-invariant acceptance set
\begin{equation}
\cA := \{X\in L^\infty \,; \ \E[X]\geq\alpha\}
\end{equation}
where $\alpha\in\R$. Since $\rho_{\cA,S}(X)=\frac{S_0}{\E[S_T]}(\alpha-\E[X])$ for any $X\in L^\infty$, the risk measure $\rho_{\cA,S}$ is law-invariant regardless of the choice of the traded asset $S$.
\end{enumerate}
\end{remark}

\medskip

The following result provides a general extension result for risk measures associated to convex, law-invariant acceptance sets. If $\cA\subset L^\infty$, we denote by $\Closure_p(\cA)$ the closure of $\cA$ in $L^p$, $1\leq p<\infty$. Note that every finite-valued risk measure $\rho_{\cA,S}$ on $L^p$, for $1\leq p\leq\infty$, is automatically continuous when $\cA$ is convex by Theorem 1 in \cite{BiaginiFrittelli2009}.

\begin{theorem}
\label{theorem on Lp}
Let $\cA\subset L^\infty$ be a convex, law-invariant acceptance set and consider a traded asset $S=(S_0,S_T)$ with $S_T\in L^\infty$. Assume $\rho_{\cA,S}$ is finite-valued and, hence, continuous on $L^\infty$. For every $1\leq p<\infty$, the following statements are equivalent:
\begin{enumerate}[(a)]
	\item $\rho_{\cA,S}$ can be extended to a finite-valued and, hence, continuous risk measure on $L^p$;
	\item if $(X_n)\subset L^\infty$ and $X_n\to0$ in $L^p$, then $\rho_{\cA,S}(X_n)\to\rho_{\cA,S}(0)$;
	\item $\Closure_\infty(\cA)$ has nonempty interior with respect to the $L^p$-topology;
	\item $\Closure_p(\cA)$ has nonempty interior in $L^p$.
\end{enumerate}
In this case, the extension is unique and is given by $\rho_{\Closure_p(\cA),S}$.
\end{theorem}
\begin{proof}
Since $\rho_{\cA,S}$ is assumed to be continuous on $L^\infty$, we have $\rho_{\cA,S}=\rho_{\Closure_\infty(\cA),S}$ by Lemma 2.5 in \cite{FarkasKochMunari2013a}. Note that $\Closure_\infty(\cA)$ is still convex and law-invariant. Indeed, $\rho_\cA$ is law-invariant and $\Closure_\infty(\cA)$ consists of all $X\in L^\infty$ such that $\rho_{\cA}(X)\leq0$. As a consequence, the theorem follows from Theorem~\ref{main theorem} and Remark~\ref{remark on law invariant sets}.
\end{proof}

\medskip

As a consequence of Theorem~\ref{theorem on Lp}, it is natural to define the index of finiteness of a risk measure $\rho_{\cA,S}$ as follows:

\begin{definition}
Let $\cA\subset L^\infty$ be a convex, law-invariant acceptance set and consider a traded asset $S=(S_0,S_T)$ with $S_T\in L^\infty$. If $\rho_{\cA,S}$ is finite-valued on $L^\infty$, the {\em index of finiteness} of $\rho_{\cA,S}$ is defined as
\begin{equation}
\label{index definition}
\Indexfin(\rho_{\cA,S}):=\inf\left\{p\in[1,\infty) \,; \ \Closure_p(\cA) \ \mbox{has nonempty interior in} \ L^p\right\}\,.
\end{equation}
\end{definition}

If the infimum in \eqref{index definition} is attained and we set $p:=\Indexfin(\rho_{\cA,S})$, then $L^p$ is the largest space for which there exists a finite-valued and, hence, continuous extension of $\rho_{\cA,S}$. Therefore, if we are interested in preserving finiteness and continuity properties of a risk measure, the space $L^p$ is to be considered the canonical model space for $\rho_{\cA,S}$. From this perspective, the canonical model space for convex, law-invariant (cash-additive) risk measures is not always $L^1$, but will depend on the individual risk measure. As illustrated below, for every $1\leq p\leq\infty$ we can find convex, law-invariant, cash-additive risk measures having index of finiteness equal to $p$. In particular, there are risk measures of this type which cannot be extended beyond $L^\infty$ in a way that preserves finiteness and continuity.

\smallskip

\begin{remark}
\label{remark index of fin}
Note that the existence of a finite-valued, continuous extension of a risk measure $\rho_{\cA,S}$ satisfying the assumptions of Theorem~\ref{theorem on Lp} does not depend on the properties of the payoff $S_T$, but only on the topological properties of the acceptance set $\cA$. An important consequence is that, if $\rho_{\cA,S}$ is finite-valued on $L^\infty$, then
\begin{equation}
\Indexfin(\rho_{\cA,S})=\Indexfin(\rho_{\cA})\,.
\end{equation}
However, the finiteness of $\rho_{\cA,S}$ on $L^\infty$ does depend on the interplay between the acceptance set and the traded asset, as illustrated by our next examples and extensively documented in \cite{FarkasKochMunari2013a}.
\end{remark}


\section{Qualitative robustness}

In this section we recall the notion of qualitative robustness introduced by Kr\"{a}tschmer, Schied, and Z\"{a}hle in \cite{KraetschmerSchiedZaehle2012} and we discuss the link with our previous results.

\medskip

Consider a law-invariant acceptance set $\cA\subset L^\infty$ and its associated cash-additive risk measure $\rho_\cA$ which is, then, also law-invariant. If we denote by $\probp_X$ the law of $X$, i.e. $\probp_X(A):=\probp(X\in A)$ for all Borel sets $A\subset\R$, and set
\begin{equation}
\cM_\infty := \{\probp_X \,; \ X\in L^\infty\}\,,
\end{equation}
we can define a functional $\cR_\cA:\cM_\infty\to\R$ by
\begin{equation}
\cR_\cA(\probp_X) := \rho_\cA(X)\,.
\end{equation}

\medskip

The capital position $X$ of a financial institution is often estimated through a sequence of historical observations $x_1,\dots,x_N\in\R$, and the quantity $\cR_\cA(m)$, where $m$ denotes the empirical distribution of these observations, is used as a natural proxy for $\rho_\cA(X)$. The importance of the robustness properties of the operator $\cR_\cA$ were discussed in detail by Cont, Deguest, and Scandolo in \cite{ContDeguestScandolo2010}. Based on that paper, a refined notion of qualitative robustness has been recently proposed in \cite{KraetschmerSchiedZaehle2012} and further studied in \cite{KraetschmerSchiedZaehle2013}.

\medskip

Let $\cM$ denote the set of (Borel) probability measures over $\R$. To any $\mu\in\cM$ which is not a Dirac measure, we can associate a nonatomic probability space $(\Omega^{\mu},\cF^{\mu},\probp^{\mu})$ supporting a sequence $(X_n)$ of i.i.d. random variables having $\mu$ as their common law, see for instance Section 11.4 in \cite{Dudley2004}. For $n\in\N$, the empirical distribution of $X_1,\dots,X_n$ is the map $m^{\mu}_n:\Omega^\mu\to\cM_\infty$ defined by
\begin{equation}
m^\mu_n(\omega):=\frac{1}{n}\sum^{n}_{i=1}\delta_{X_i(\omega)} \ \ \ \mbox{for} \ \omega\in\Omega^\mu\,,
\end{equation}
where $\delta_{X_i(\omega)}$ denotes the standard Dirac measure associated to the singleton $\{X_i(\omega)\}$. Moreover, we can consider the random variable $\cR_\cA(m^\mu_n)$ given by
\begin{equation}
\cR_\cA(m^\mu_n)(\omega) := \cR_\cA(m^\mu_n(\omega)) \ \ \ \mbox{for} \ \omega\in\Omega^\mu\,.
\end{equation}

\smallskip

The following notion of qualitative robustness is a generalization of the classical notion introduced by Hampel in \cite{Hampel1971}. For $1\leq p<\infty$ define $\psi_p(x):=\frac{1}{p}\left|x\right|^p$, $x\in\R$, and recall from \cite{KraetschmerSchiedZaehle2013} that a set $\cN\subset\cM$ is said to be \textit{uniformly $p$-integrating} if
\begin{equation}
\lim_{M\to\infty}\sup_{\mu\in\cN}\int_{\{\psi_p\geq M\}}\psi_p(x)d\mu(x)=0\,.
\end{equation}

\begin{definition}
The functional $\cR_\cA$ is said to be {\em $p$-robust on $\cM_\infty$}, $1\leq p<\infty$, if for any uniformly $p$-integrating set $\cN\subset\cM_\infty$, $\mu\in\cN$ and $\e>0$ there exist $\delta>0$ and $n_0\in\N$ such that
\begin{equation}
\label{robustness formula 1}
d_P(\mu,\nu)+\left|\int_\R\psi_p(x)d\mu(x)-\int_\R\psi_p(x)d\nu(x)\right|\leq\delta
\end{equation}
implies
\begin{equation}
\label{robustness formula 2}
d_P\left(\probp^\mu_{\cR_\cA(m^\mu_n)},\probp^\nu_{\cR_\cA(m^\nu_n)}\right)\leq\e
\end{equation}
for $\nu\in\cN$ and $n\geq n_0$, where $d_P$ denotes the usual Prohorov metric over $\cM$.
\end{definition}

Hence, if $\cR_\cA$ is $p$-robust on $\cM_\infty$, then a suitable small change in the law of the data entails an arbitrarily small change in the law of the corresponding estimators.

\smallskip

\begin{remark}
As discussed in \cite{KraetschmerSchiedZaehle2012} and \cite{KraetschmerSchiedZaehle2013}, the choice to add an additional term to the Prohorov metric in \eqref{robustness formula 1}, as opposed to the classical framework developed by Hampel in \cite{Hampel1971}, has the main advantage of making $\cR_\cA(\mu)$ sensitive to the tail behaviour of $\mu$. Indeed, under the Prohorov metric, or equivalently under any metric inducing the weak topology on $\cM$, like the L\'{e}vy metric, two distributions $\mu$ and $\nu$ may possess a different tail behaviour but be rather close in metric terms. In this case, qualitative robustness would essentially prevent $\cR_\cA$ from discriminating across different tail profiles. For more details about the Prohorov and L\'{e}vy metric we refer to Section 11.3 in \cite{Dudley2004}.
\end{remark}

\medskip

Based on \cite{KraetschmerSchiedZaehle2012}, the same authors introduced in \cite{KraetschmerSchiedZaehle2013} the \textit{index of qualitative robustness} for a risk measure $\rho_\cA$ defined by
\begin{equation}
\Indexiqr(\rho_\cA) := \left(\inf\{p\in[1,\infty) \,; \ \cR_\cA \ \mbox{is} \ p\mbox{-robust on} \ \cM_\infty\}\right)^{-1}\,.
\end{equation}

By combining Theorem 2.16 in \cite{KraetschmerSchiedZaehle2013} and our previous Theorem \ref{theorem on Lp} we obtain the following interesting result relating the qualitative robustness of the operator $\cR_\cA$ to the topological properties of the acceptance set $\cA$.

\begin{theorem}
Assume $\cA\subset L^\infty$ is a convex, law-invariant acceptance set, and let $1\leq p<\infty$. The following statements are equivalent:
\begin{enumerate}[(a)]
	\item $\cR_\cA$ is $p$-robust on $\cM_\infty$;
	\item $\Closure_p(\cA)$ has nonempty interior in $L^p$.
\end{enumerate}
Moreover, we have
\begin{equation}
\label{index robustness}
\Indexiqr(\rho_\cA)=\frac{1}{\Indexfin(\rho_\cA)}\,.
\end{equation}
\end{theorem}

\medskip

In the final sections we compute the index of finiteness of several risk measures. As a consequence of the above theorem, these examples turn out to be important also from the perspective of qualitative robustness.


\section{Risk measures based on utility functions}

In this section we analyse the index of finiteness of risk measures based on expected utility. Note that, even though such risk measures are treated in \cite{KraetschmerSchiedZaehle2013}, no results concerning their statistical robustness are proved there.

\medskip

Recall that a nonconstant function $u:\R\to\R\cup\{-\infty\}$ is said to be a \textit{utility function} if $u$ is increasing and concave. This implies that $u$ is unbounded from below. In the sequel, we assume that $u$ denotes a utility function which is \textit{bounded from above}.

\medskip

For every $1\leq p\le\infty$ and a level $\alpha\in\R$ we set
\begin{equation}
\cA_u^p:=\{X\in L^p \,; \ \E[u(X)]\geq\alpha\}\,.
\end{equation}

Clearly, this set is nonempty if and only if $u(x)\ge \alpha$ for some $x\in\R$, which we assume from now on. Moreover, in that case $\cA_u^p$ is a convex, law-invariant acceptance set.

\medskip

We start by providing a characterization of when risk measures of the form $\rho_{\cA_u^\infty,S}$ are finite-valued on $L^\infty$.

\begin{proposition}
Let $S=(S_0,S_T)$ be a traded asset with $S_T\in L^\infty$.
\begin{enumerate}[(i)]
	\item Assume $u$ never attains the value $-\infty$ and $u(x)>\alpha$ for some $x\in\R$. Then the following are equivalent:
\begin{enumerate}[(a)]
	\item $\rho_{\cA_u^\infty,S}$ is finite-valued and, hence, continuous on $L^\infty$;
	\item $\probp(S_T=0)=0$.
\end{enumerate}
	\item Assume $u$ attains the value $-\infty$ or $u(x)\leq\alpha$ for all $x\in\R$. Then the following are equivalent:
\begin{enumerate}[(a)]
	\item $\rho_{\cA_u^\infty,S}$ is finite-valued and, hence, continuous on $L^\infty$;
	\item $\probp(S_T\geq\e)=1$ for some $\e>0$.
\end{enumerate}	
\end{enumerate}
\end{proposition}
\begin{proof}
First, we show that $\rho_{\cA_u^\infty,S}$ never takes the value $-\infty$. To this end, fix $X\in L^\infty$ and $\gamma>0$ such that $\probp(S_T\geq\gamma)>0$. Then, since $u$ is unbounded from below, we can always find $\lambda>0$ sufficiently large to yield
\begin{equation}
\E[u(X-\lambda S_T)] \leq u(\left\|X\right\|_\infty-\lambda\gamma)\,\probp(S_T\geq\gamma)+\sup_{x\in\R}u(x)\,
\probp(S_T<\gamma) < \alpha\,.
\end{equation}
This implies $X-\lambda S_T\notin\cA_u^\infty$ and, hence, $\rho_{\cA_u^\infty,S}(X)>-\infty$.

\smallskip

To prove~\textit{(i)}, assume first that \textit{(a)} holds so that $\rho_{\cA_u^\infty,S}(-\xi 1_\Omega)<\infty$ for any $\xi>0$. As a result, for every $\xi>0$ there exists $\lambda>0$ such that
\begin{equation}
u(-\xi)\,\probp(S_T=0)+\sup_{x\in\R}u(x)\,\probp(S_T>0) \geq \E[u(-\xi 1_\Omega+\lambda S_T)] \geq \alpha\,.
\end{equation}
Since $u$ is unbounded below, this is only possible if $\probp(S_T=0)=0$, proving \textit{(b)}.

\smallskip

Now assume \textit{(b)} holds and take $X\in L^\infty$. Since $u(x)>\alpha$ for some $x\in\R$ and $\probp(S_T=0)=0$, we can find $\e>0$ sufficiently small to obtain
\begin{equation}
\E\left[u\left(X+\frac{1}{\e^2}S_T\right)\right] \geq u\left(\frac{1}{\e}-\left\|X\right\|_\infty\right)\,\probp(S_T\geq\e)+u(-\left\|X\right\|_\infty)\,\probp(S_T<\e) \geq \alpha\,,
\end{equation}
implying that $\rho_{\cA_u^\infty,S}(X)<\infty$. Hence, \textit{(a)} follows since $\rho_{\cA_u^\infty,S}$ never attains the value $-\infty$.

\smallskip

To prove~\textit{(ii)}, we first show that~\textit{(b)} always implies~\textit{(a)}. Indeed, if~\textit{(b)} holds then $S_T$ is an interior point of $L^\infty$, hence $\rho_{\cA_u^\infty,S}$ is finite-valued by Proposition 3.1 in \cite{FarkasKochMunari2013a}.

\smallskip

Conversely, assume that~\textit{(a)} holds under the condition that $u(-\xi)=-\infty$ for some $\xi>0$. In this case, set $X:=(-\xi-1)1_\Omega$ and note that for every $\lambda>0$ there exists $\e>0$ such that $u(-\xi-1+\lambda\e)=-\infty$. Now, if $\probp(S_T<\e)>0$ for all $\e>0$, this implies
\begin{equation}
\E[u(X+\lambda S_T)] \leq u(-\xi-1+\lambda\e)\,\probp(S_T<\e)+\sup_{x\in\R}u(x)\,\probp(S_T\geq\e) < \alpha\,.
\end{equation}
As a result $\rho_{\cA_u^\infty,S}(X)=\infty$, contradicting~\textit{(a)}. Hence, we must have $\probp(S_T\geq\e)>0$ for some $\e>0$ so that~\textit{(b)} holds.

\smallskip

Finally, assume that~\textit{(a)} holds and $u$ is bounded from above by $\alpha$, and set $x_0:=\inf\{x\in\R \,; \ u(x)=\alpha\}$. Moreover, take $\xi>-x_0$. Since $\rho_{\cA_u,S}(-\xi1_\Omega)<\infty$, there exists $\lambda>0$ such that $\E[u(-\xi1_\Omega+\lambda S_T)]\geq\alpha$. But this is only possible if $-\xi+\lambda S_T\geq x_0$ almost surely, implying that $\probp(S_T\geq\frac{\xi+x_0}{\lambda})=1$. As a consequence~\textit{(b)} holds, concluding the proof.
\end{proof}

\medskip

To study extension properties of risk measures based on expected utility, we first need to investigate the topological structure of the corresponding acceptance sets.

\begin{lemma}
\label{lemma Au closed}
For every $1\leq p\le\infty$, the acceptance set $\cA_u^p$ is closed in $L^p$.
\end{lemma}
\begin{proof}
To prove that $\cA_u^p$ is closed in $L^p$, take a sequence $(X_n)$ in $\cA_u^p$ and assume $X_n\to X$ in $L^p$ as $n\to\infty$. Since $X_n\to X$ almost surely as $n\to\infty$, up to passing to a suitable subsequence, it follows from the continuity of $u$ and by Fatou's lemma, e.g. Lemma 4.3.3 in~\cite{Dudley2004}, that
\begin{equation}
\E[u(X)] = \E\left[\lim u(X_n)\right] \geq \limsup \E[u(X_n)] \geq \alpha\,.
\end{equation}
This shows that $X\in\cA_u^p$ and, hence, that $\cA_u^p$ is closed.
\end{proof}

\medskip

Assume $S=(S_0,S_T)$ is a traded asset with $S_T\in L^\infty_+$ such that $\rho_{\cA,S}$ is finite-valued and fix $1\leq p<\infty$. Then by Theorem~\ref{theorem on Lp} we know that $\rho_{\cA,S}$ can be extended to a finite-valued, continuous risk measure on $L^p$ if and only if $\Closure_p(\cA_u^\infty)$ has nonempty interior in $L^p$. Since, by the above lemma, $\Closure_p(\cA_u^\infty)\subset\cA_u^p$ we infer that if $\cA_u^p$ has empty interior, then $\rho_{\cA,S}$ cannot admit such an extension. The result below provides conditions for $\cA_u^p$ to have empty interior. In particular, condition \textit{(ii)} shows that this may depend on the decay behaviour of the utility function at $-\infty$.

\begin{lemma}
\label{lemma interior utility}
Fix $1\leq p<\infty$ and assume one of the following conditions:
\begin{enumerate}[(i)]
  \item $u(x)\leq\alpha$ for all $x\in\R$;
  \item $\lim_{x\to\infty}\frac{x^p}{u(-x)}=0$.
\end{enumerate}
Then $\cA_u^p$ has empty interior in $L^p$.
\end{lemma}
\begin{proof}
\textit{(i)} Take $X\in\cA_u^p$ and $r>0$. Choose first $\gamma>0$ with
$\probp(\left|X\right|<\gamma)>0$ and then $\xi>0$ such that $u(\gamma -\xi)<\alpha$. Since $(\Omega,\cF,\probp)$ is nonatomic we find $A\subset\{\vert X\vert <\gamma\}$ with $\probp(A)<\frac{r^p}{\xi^p}$. Set now $Y:=(X-\xi)1_A+X1_{A^c}$, and note that $\left\|X-Y\right\|^p_p=\xi^p\probp(A)<r^p$. Moreover,
\begin{equation}
\E[u(Y)] = \E[u(X-\xi)1_A] + \E[u(X)1_{A^c}] \leq u(\gamma-\xi)\,\probp(A)+\alpha\,\probp(A^c) < \alpha\,.
\end{equation}
Hence, in every neighborhood of $X$ there exists some element which does not belong to $\cA_u^p$. Since $X$ was arbitrary, this implies $\cA_u^p$ has empty interior.

\smallskip

\textit{(ii)} Take $X\in\cA_{u}^p\cap L^\infty$ so that $u(\left\|X\right\|_\infty)\geq\E[u(X)]\geq\alpha$, and fix $r>0$. It is easy to see that by assumption we can find a sufficiently large $\xi>0$ such that
\begin{equation}
0 \leq \frac{u(\left\|X\right\|_\infty)-\alpha}{u(\left\|X\right\|_\infty)-u(\left\|X\right\|_\infty-\xi)} < \frac{r^p}{
\xi^p} < 1\,.
\end{equation}
As a consequence, taking $\lambda\in(0,1)$ with
\begin{equation}
\frac{u(\left\|X\right\|_\infty)-\alpha}{u(\left\|X\right\|_\infty)-u(\left\|X\right\|_\infty-\xi)} < \lambda < \frac{r^p}{\xi^p}
\end{equation}
we obtain $\xi^p \lambda<r^p$ and
\begin{equation}
\label{ineq for empty int in Lp}
\lambda u(\left\|X\right\|_\infty-\xi) + (1-\lambda)u(\left\|X\right\|_\infty) < \alpha\,.
\end{equation}
Since $(\Omega,\cF,\probp)$ is nonatomic, $\probp(A)=\lambda$ for a suitable $A\in\cF$. Now, consider the random variable $Y:=(X-\xi)1_A+X1_{A^c}$. Clearly, $\left\|X-Y\right\|^p_p = \xi^p \probp(A)<r^p$. Moreover, as a consequence of \eqref{ineq for empty int in Lp}, we obtain
\begin{equation}
\E[u(Y)] \leq \probp(A)u(\left\|X\right\|_\infty-\xi) + \probp(A^c)u(\left\|X\right\|_\infty) < \alpha\,.
\end{equation}
This implies that $X$ is not an interior point of $\cA_{u}^p$. As a result, by the density of $L^\infty$ in $L^p$ we can conclude that $\cA_{u}^p$ has empty interior.
\end{proof}

\medskip

The following result follows immediately from the discussion preceding Lemma~\ref{lemma interior utility}.

\begin{corollary}
\label{corollary infinite index}
Assume that either $u(x)\le \alpha$ for all $x\in\R$ or that $u$ attains the value $-\infty$. Then, for any traded asset $S=(S_0,S_T)$ with $S_T\in L^\infty$ such that $\rho_{\cA,S}$ is finite-valued on $L^\infty$, we have $\Indexfin(\rho_{\cA_u^\infty,S})=\infty$.
\end{corollary}

\smallskip

\begin{remark}
As an example of a utility function attaining the value $-\infty$ we can consider a capped \textit{log-utility} of the form
\begin{equation}
u(x):=\left\{
\begin{array}{c l}
C & \mbox{if} \ x\geq c \\
\log(1+x) & \mbox{if} \ 0\leq x<c \\
-\infty & \mbox{if} \ x<0
\end{array}
\right.
\end{equation}
for fixed constants $c>0$ and $C=\log(1+c)$.
\end{remark}

\medskip

In view of Corollary~\ref{corollary infinite index} we assume for the rest of this section that \textit{$u$ is finite-valued and $u(x)>\alpha$ for some $x\in\R$}. Under this assumption, we can refine Theorem~\ref{theorem on Lp} as follows.

\begin{theorem}
\label{corollary utility on Lp}
\begin{enumerate}[(i)]
\item For any $1\leq p<\infty$ we have $\Closure_p(\cA_u^\infty)=\cA_u^p$.
\item Let $S=(S_0,S_T)$ be a traded asset with $S_T\in L^\infty$. Assume $\rho_{\cA_u^\infty,S}$ is finite-valued and, hence, continuous on $L^\infty$, and fix $1\leq p<\infty$. The following statements are equivalent:
\begin{enumerate}[(a)]
	\item $\rho_{\cA_u^\infty,S}$ can be extended to a finite-valued and, hence, continuous risk measure on $L^p$;
	\item $\cA_u^p$ has nonempty interior in $L^p$.
\end{enumerate}
In this case, the extension is unique and given by $\rho_{\cA_u^p,S}$.
\end{enumerate}
\end{theorem}
\begin{proof}
By virtue of Theorem~\ref{theorem on Lp} it is enough to show part \textit{(i)}. To this end, since $\cA_u^p$ is closed by Lemma~\ref{lemma Au closed}, we only need to prove that any element $X\in\cA_u^p$ is the limit in $L^p$ of a suitable sequence $(X_n)$ of elements in $\cA_u^\infty$. Now, take $X\in\cA_u^p$.

\smallskip

Since there exists $x\in\R$ such that $u(x)>\alpha$, we can find $Y\in L^p$ with $\E[u(Y)]>\alpha$. Then, setting $Z_\lambda:=\lambda X+(1-\lambda)Y$ for $\lambda\in(0,1)$, the concavity of $u$ yields
\begin{equation}
\E[u(Z_\lambda)] \geq \lambda\E[u(X)]+(1-\lambda)\E[u(Y)] > \alpha\,.
\end{equation}
Since $Z_\lambda\to X$ in $L^p$ as $\lambda\to 1$, this shows we may assume that $\E[u(X)]>\alpha$ without loss of generality.

\smallskip

Now, assume $X$ is bounded from below almost surely, and set $X_n:=X1_{\{X\le n\}}\in L^\infty$ for any $n\in\N$. Then
\begin{equation}
\label{auxiliary equation}
\alpha < \E[u(X)] = \E[u(X_n)]+\E[u(X1_{\{X>n\}})]\,.
\end{equation}
Since $u$ is bounded from above, we have $\E[u(X1_{\{X>n\}})]\to 0$ as $n\to\infty$ by dominated convergence, hence $\E[u(X_n)]>\alpha$ for large enough $n\in\N$. In particular, we eventually have $X_n\in\cA_u^\infty$. This shows that $X\in\Closure_p(\cA_u^\infty)$ because $X_n\to X$ in $L^p$ as $n\to\infty$.

\smallskip

Finally, assume $X$ is not bounded from below almost surely and define for each $n\in\N$ the random variable $X_n:=X1_{\{X\ge -n\}}\in L^p$. Clearly, $X_n\to X$ in $L^p$ as $n\to\infty$. Moreover, $\E[u(X_n)]\ge\E[u(X)]>\alpha$ for all $n\in\N$ by the monotonicity of $u$. Since every $X_n$ is bounded from below almost surely, we can rely on the previous argument and conclude that $X_n\in\Closure_p(\cA_u^\infty)$ for any $n\in\N$ so that $X\in\Closure_p(\cA_u^\infty)$.
\end{proof}


\subsubsection*{Exponential utility}

The index of finiteness may be $\infty$ even if $u$ never attains the value $-\infty$. To see this we consider the \textit{exponential utility} function $u(x):=1-e^{-\gamma x}$, $x\in\R$, for some fixed $\gamma>0$. The following result shows that finite-valued risk measures on $L^\infty$ based on expected exponential utility do not admit finite-valued, hence continuous, extensions to any $L^p$ space, $1\le p<\infty$.

\begin{corollary}
Let $S=(S_0,S_T)$ be a traded asset with $S_T\in L^\infty$, and assume $\rho_{\cA_u^\infty,S}$ is finite-valued. Then $\Indexfin(\rho_{\cA_u^\infty,S})=\infty$.
\end{corollary}
\begin{proof}
For any  $1\leq p<\infty$ we have
\begin{equation}
\lim_{x\to\infty}\frac{x^p}{u(-x)}=\lim_{x\to\infty}\frac{x^p}{1-e^{\gamma x}}=0\,.
\end{equation}
Hence, Lemma \ref{lemma interior utility} implies that the interior of $\cA_u^p$ is empty, thus $\rho_{\cA_u^\infty,S}$ does not admit any finite-valued, continuous extension to $L^p$ by Theorem~\ref{corollary utility on Lp}.
\end{proof}


\subsubsection*{Flat power utility}

We now show that we can find convex risk measures on $L^\infty$ whose index of finiteness is equal to any prescribed number $1\leq q<\infty$. To this effect recall that the \textit{flat power utility} function is defined by
\begin{equation}
u(x):=\left\{
\begin{array}{c l}
-\left|x\right|^q &\mbox{if} \ x<0 \\
0 &\mbox{if} \ x\geq0
\end{array}
\right.
\end{equation}
where $1\leq q<\infty$.

\begin{corollary}
Let $S=(S_0,S_T)$ be a traded asset with $S_T\in L^\infty$, and assume $\rho_{\cA_u^\infty,S}$ is finite-valued. Then $\Indexfin(\rho_{\cA_u^\infty,S})=q$ and the index is attained.
\end{corollary}
\begin{proof}
First, note that
\begin{equation}
\E[u(X)]=-\left\|X\wedge0\right\|^q_q \ \ \ \mbox{for all} \ X\in L^1\,.
\end{equation}
Since we assumed that $u(x)>\alpha$ for some $x\in\R$, this implies $\alpha<0$ in the present case.

\smallskip

For $p\geq q$ the map $U:L^p\to\R$ defined by
\begin{equation}
U(X):=\E[u(X)] \ \ \ \mbox{for} \ X\in L^p
\end{equation}
is easily seen to be continuous. Since $\cA_u^p$ contains the nonempty, open set $U^{-1}((\alpha,\infty))$, it must have nonempty interior, hence $\Indexfin(\rho_{\cA_u^\infty,S})\leq q$ by Theorem~\ref{corollary utility on Lp}. In particular, note that $\rho_{\cA_u^\infty,S}$ can be extended to a finite-valued, continuous risk measure on $L^q$.

\smallskip

If $p<q$, then it is immediate to see that
\begin{equation}
\lim_{x\to\infty}\frac{x^p}{u(-x)}=-\lim_{x\to\infty}x^{p-q}=0\,.
\end{equation}
Consequently, the interior of $\cA_u^p$ is empty by Lemma \ref{lemma interior utility}, hence $\Indexfin(\rho_{\cA_u^\infty,S})\geq q$ as a consequence of Theorem~\ref{corollary utility on Lp}. In conclusion, $\Indexfin(\rho_{\cA_u^\infty,S})=q$ and the index is attained.
\end{proof}


\subsubsection*{An example of a non-HARA utility}

In this section we focus on the utility function
\begin{equation}
u(x):=\left\{
\begin{array}{c l}
C &\mbox{if} \ x\geq c \\
\frac{1}{a}(1+ax-\sqrt{1+a^2x^2}) &\mbox{if} \ x<c
\end{array}
\right.
\end{equation}
for fixed parameters $a>0$ and $c\geq0$, and $C=\frac{1}{a}(1+ac-\sqrt{1+a^2c^2})$. The uncapped version was proposed in Section 2.2.2 in \cite{Carmona2009} as a tractable alternative to exponential utility if one wants to penalize negative wealth less severely. The following result shows that the corresponding risk measures can always be extended to $L^1$.

\begin{corollary}
Let $S=(S_0,S_T)$ be a traded asset with $S_T\in L^\infty$, and assume $\rho_{\cA_u^\infty,S}$ is finite-valued. Then $\Indexfin(\rho_{\cA_u^\infty,S})=1$ and the index is attained.
\end{corollary}
\begin{proof}
Define the map $U:L^1\to\R$ by setting
\begin{equation}
U(X) := \E[u(X)] \ \ \ \mbox{for} \ X\in L^1\,.
\end{equation}
Since $U$ is concave and increasing, it is continuous by Theorem 1 in \cite{BiaginiFrittelli2009}. As a result, $\cA^1_u$ has nonempty interior because it contains the nonempty, open set $U^{-1}((\alpha,\infty))$. In conclusion, Theorem~\ref{corollary utility on Lp} implies that $\Indexfin(\rho_{\cA_u^\infty,S})=1$ and the index is clearly attained.
\end{proof}


\section{Max-correlation risk measures}

In this section we provide a characterization of the index of finiteness for the so-called max-correlation risk measure introduced by R\"{u}schendorf in~\cite{Rueschendorf2006} and studied by Ekeland and Schachermayer in~\cite{EkelandSchachermayer2011} and by Ekeland, Galichon, and Henry in~\cite{EkelandGalichonHenry2012}.

\medskip

Consider a probability measure $\probq$ on $(\Omega,\cF)$ that is absolutely continuous with respect to $\probp$. Assume that $1\leq p\leq\infty$ is such that $\frac{d\probq}{d\probp}\in L^{p^\prime}$ and define the \textit{max-correlation} risk measure $\rho_{\probq,p}:L^p\to\R\cup\{\infty\}$ by
\begin{equation}
\rho_{\probq,p}(X):=\sup\left\{\E[-XY] \,; \ Y\sim\frac{d\probq}{d\probp}\right\}\quad\mbox{for} \ X\in L^p\,.
\end{equation}
As a consequence of Theorem~13.4 in~\cite{ChongRice1971}, for any $X\in L^\infty$ we have the equivalent (and more common) formulation
\begin{equation}
\rho_{\probq,p}(X):=\sup_{X'\sim X}\E_\probq[-X'] \ \ \ \mbox{for} \ X\in L^p\,.
\end{equation}

The acceptance set associated with $\rho_{\probq,p}$ is given by
\begin{equation}
\cA_\probq^p:=\{X\in L^p \,; \ \rho_{\probq,p}(X)\leq0\}=\left\{X\in L^p \,; \ \E[XY]\geq0, \ \forall Y\sim\frac{d\probq}{d\probp}\right\}\,.
\end{equation}
Clearly, $\cA_\probq^p$ is law-invariant and coherent, i.e. a convex cone.

\medskip

We start by showing when the risk measure $\rho_{\cA_\probq^\infty,S}$ is finite-valued on $L^\infty$.

\begin{proposition}
Let $S=(S_0,S_T)$ be a traded asset with $S_T\in L^\infty$. The following are equivalent:
\begin{enumerate}[(a)]
	\item $\rho_{\cA_\probq^\infty,S}$ is finite-valued and, hence, continuous on $L^\infty$;
	\item $\inf_{Z\sim S_T}\E_\probq[Z]>0$.
\end{enumerate}
\end{proposition}
\begin{proof}
Since $\cA_\probq^\infty$ is coherent, it follows from Proposition 3.6 and Theorem 3.16 in \cite{FarkasKochMunari2013a} that \textit{(a)} is equivalent to $S_T$ being an interior point of $\cA_\probq^\infty$. By the continuity of the cash-additive risk measure $\rho_{\probq,\infty}$, this is equivalent to $\rho_{\probq,\infty}(S_T)<0$, concluding the proof.
\end{proof}

\medskip

Before proving the extension result for max-correlation risk measures, we need the following lemma.

\begin{lemma}
Let $1\leq p\leq\infty$ and assume that $ \frac{d\probq}{d\probp}\in L^{p^\prime}$. Then $\rho_{\probq,p}$ is finite-valued and, hence, continuous on $L^p$.
\end{lemma}
\begin{proof}
For $X\in L^p$ we have
\begin{equation}
\left|\rho_{\probq,p}(X)\right| \leq \left\|X\right\|_p\left\|\frac{d\probq}{d\probp}\right\|_{p'}
\end{equation}
so that $\rho_{\probq,p}$ is finite-valued on $L^p$. Moreover, for any $X,Y\in L^p$ we have $\rho_{\probq,p}(X)\le \rho_{\probq,p}(X-Y)+\rho_{\probq,p}(Y)$ by subadditivity and, consequently,
\begin{equation}
\label{bound for maxcorr}
\left|\rho_{\probq,p}(X)-\rho_{\probq,p}(Y)\right| \le \left\|X-Y\right\|_p \left\| \frac{d\probq}{d\probp}\right\|_{p'}\,.
\end{equation}
It follows that $\rho_{\probq,p}$ is Lipschitz-continuous on $L^p$.
\end{proof}

\medskip

We now characterize for which $1\leq p<\infty$ the risk measure $\rho_{\probq,\infty}$ admits a finite-valued, continuous extension to $L^p$.

\begin{proposition}
\label{maxcorr prop}
For $1\leq p<\infty$ the following holds:
\begin{enumerate}[(i)]
\item If $\frac{d\probq}{d\probp}\in L^{p'}$, then $\rho_{\probq,\infty}$ admits a unique finite-valued, continuous extension to $L^p$ which is given by $\rho_{\probq,p}$.
\item If $\frac{d\probq}{d\probp}\not\in L^{p'}$, then $\rho_{\probq,\infty}$ does not admit finite-valued, continuous extensions to $L^p$.
\end{enumerate}
\end{proposition}
\begin{proof}
Since \textit{(i)} follows readily from the preceding lemma, we only need to prove \textit{(ii)}. Assume that $\rho_{\probq,\infty}$ admits a finite-valued and, hence continuous extension to $L^p$. Since $\cA_\probq^\infty$ is closed, Theorem~\ref{theorem on Lp} implies that it must have nonempty interior with respect to the $L^p$-topology. Consider now the linear functional $V:L^\infty\to\R$ defined by
\begin{equation}
V(X):=\E_\probq[X] \ \ \ \mbox{for} \ X\in L^\infty\,.
\end{equation}
Note that $\cA^\infty_\probq\subset V^{-1}([0,\infty))$ implies that $V^{-1}([0,\infty))$ has nonempty interior with respect to the $L^p$-topology. Therefore, $V$ is continuous with respect to that topology. As a result, there exists a continuous, linear functional $\overline{V}:L^p\to\R$ extending $V$. In particular, we can find $Z\in L^{p'}$ such that
\begin{equation}
\E[XZ]=\overline{V}(X)=V(X)=\E_\probq[X] \ \ \ \mbox{for all} \ X\in L^\infty\,,
\end{equation}
implying $\frac{d\probq}{d\probp}=Z$ almost surely. Hence, $\frac{d\probq}{d\probp}\in L^{p^\prime}$ contradicting the assumption. Consequently, $\rho_{\probq,\infty}$ does not admit any finite-valued and, hence, continuous extension to $L^p$.
\end{proof}

\medskip

Set now
\begin{equation}
q:=\sup\left\{p'\in[1,\infty) \,; \ \frac{d\probq}{d\probp}\in L^{p'}\right\}\,.
\end{equation}

The following result characterizes the index of finiteness of $\rho_{\cA_\probq^\infty,S}$. For $\rho_{\probq,\infty}$, it is an immediate consequence of Proposition~\ref{maxcorr prop}. The extension in the case of a general traded asset is ensured by Remark~\ref{remark index of fin}.

\begin{corollary}
\label{corollary on maxcorr}
Let $S=(S_0,S_T)$ be a traded asset with $S_T\in L^\infty$, and assume $\rho_{\cA_\probq^\infty,S}$ is finite-valued. Then $\Indexfin(\rho_{\cA_\probq^\infty,S})=q'$ and the index is attained if and only if $\frac{d\probq}{d\probp}\in L^q$.
\end{corollary}

\smallskip

\begin{remark}
It is known that the max-correlation risk measure is a distortion risk measure, see e.g. Remark 2.6 in \cite{Rueschendorf2006}. Therefore, an alternative strategy to prove Corollary~\ref{corollary on maxcorr} would be to use the results in the next section. However, the above proof is more direct and simpler.
\end{remark}


\section{Distortion risk measures}

In this section we rely on the results for cash-additive distortion risk measures obtained in \cite{KraetschmerSchiedZaehle2013} and derive the corresponding index of finiteness for general risk measures which need not be cash-additive.

\medskip

Let $\delta:[0,1]\to[0,1]$ be a concave, increasing function satisfying $\delta(0)=0$ and $\delta(1)=1$. For $X\in L^\infty$ we denote by $F_X$ the distribution function of $X$. The corresponding \textit{distortion risk measure} is the map $\rho_\delta:L^\infty\to\R$ defined by
\begin{equation}
\rho_\delta(X) := \int^0_{-\infty}\delta(F_X(x))dx-\int_0^\infty(1-\delta(F_X(x)))dx \ \ \ \mbox{for} \ X\in L^\infty\,.
\end{equation}
We refer to Section 4.6 in \cite{FoellmerSchied2011} for more details about this type of risk measures. As it is well-known, $\rho_\delta$ is a coherent, law-invariant, cash-additive risk measure, hence the corresponding acceptance set
\begin{equation}
\cA_\delta := \{X\in L^\infty \,; \ \rho_\delta(X)\leq0\}
\end{equation}
is law-invariant and coherent.

\medskip

First, we characterize when general risk measures associated to the acceptance set $\cA_\delta$ are finite-valued on $L^\infty$.

\begin{proposition}
Let $S=(S_0,S_T)$ be a traded asset with $S_T\in L^\infty$. The following are equivalent:
\begin{enumerate}[(a)]
  \item $\rho_{\cA_\delta,S}$ is finite-valued and, hence, continuous on $L^\infty$;
  \item $\delta(F_{S_T}(x))<1$ for some $x>0$.
\end{enumerate}
In particular, if $\delta$ is strictly increasing on some left neighborhood of $1$, then \textit{(a)} holds.
\end{proposition}
\begin{proof}
First, note that $\cA_\delta$ has nonempty interior in $L^\infty$ because it contains a translate of $L^\infty_+$ by monotonicity, and the corresponding interior points are those $X\in L^\infty$ satisfying $\rho_{\cA_\delta}(X)<0$. By combining Proposition 3.6 and Theorem 3.16 in \cite{FarkasKochMunari2013a}, it follows that $\rho_{\cA_\delta,S}$ is finite-valued on $L^\infty$ if and only if $S_T$ belongs to the interior of $\cA_\delta$. As a result, the assertion \textit{(a)} is then equivalent to
\begin{equation}
\rho_{\cA_\delta}(S_T) = -\int_0^\infty(1-\delta(F_{S_T}(x)))dx < 0\,,
\end{equation}
which, in turn, is equivalent to \textit{(b)} by virtue of the monotonicity of $\delta$.
\end{proof}

\medskip

Now, define
\begin{equation}
q := \sup\left\{p\in[1,\infty) \,; \ \int_0^1(\delta'_+(\lambda))^p d\lambda<\infty\right\}
\end{equation}
where $\delta'_+$ denotes the right derivative of $\delta$.

\smallskip

The following result follows directly from Proposition 2.22 in \cite{KraetschmerSchiedZaehle2013} combined with Remark \ref{remark index of fin}.

\begin{proposition}
Let $S=(S_0,S_T)$ be a traded asset with $S_T\in L^\infty$, and assume $\rho_{\cA_\delta,S}$ is finite-valued. Then $\Indexfin(\rho_{\cA_\delta,S})=q'$ and the index is attained if and only if $\int_0^1(\delta'_+(\lambda))^q d\lambda<\infty$.
\end{proposition}

\smallskip

\begin{example}
Let $S=(S_0,S_T)$ be a traded asset with $S_T\in L^\infty$, and assume the corresponding risk measure $\rho_{\cA_\delta,S}$ is finite-valued over $L^\infty$. The following distortion functions are discussed in \cite{ChernyMadan2010}; see also \cite{KraetschmerSchiedZaehle2013}.

\smallskip

The MAXVAR risk measure corresponds to the distortion function
\begin{equation}
\delta(x)=x^{\frac{1}{\gamma}} \ \ \ \mbox{for} \ x\in[0,1] \ \mbox{and} \ \gamma\geq1\,.
\end{equation}
A direct computation shows that $\Indexfin(\rho_{\cA_\delta,S})=\gamma$ and that the index is not attained.

\smallskip

Similarly, the MINVAR risk measure corresponds to
\begin{equation}
\delta(x)=1-(1-x)^\gamma \ \ \ \mbox{for} \ x\in[0,1] \ \mbox{and} \ \gamma\geq1\,.
\end{equation}
In this case, $\Indexfin(\rho_{\cA_\delta,S})=1$ and the index is attained.

\smallskip

The MAXMINVAR risk measure is associated to the distortion
\begin{equation}
\delta(x)=(1-(1-x)^\gamma)^{\frac{1}{\gamma}} \ \ \ \mbox{for} \ x\in[0,1] \ \mbox{and} \ \gamma\geq1\,.
\end{equation}
Since $\delta'(x)\sim(\gamma x)^{\frac{1}{\gamma}-1}$ for $x\to0$, it follows that $\Indexfin(\rho_{\cA_\delta,S})=\gamma$ and the index is not attained.

\smallskip

Similarly, the MINMAXVAR risk measure corresponding to
\begin{equation}
\delta(x)=\left(1-(1-x)^{\frac{1}{\gamma}}\right)^\gamma \ \ \ \mbox{for} \ x\in[0,1] \ \mbox{and} \ \gamma\geq1
\end{equation}
is such that $\Indexfin(\rho_{\cA_\delta,S})=\gamma$, and the index is not attained.

\smallskip

We can also consider the distortion
\begin{equation}
\delta(x)=\left(1-(1-x)^{\frac{1}{\beta}}\right)^\gamma \ \ \ \mbox{for} \ x\in[0,1] \ \mbox{and} \ \beta,\gamma\geq1\,.
\end{equation}
In this case $\Indexfin(\rho_{\cA_\delta,S})=\beta$, in accordance with Example 2.23 in \cite{KraetschmerSchiedZaehle2013}, and the index is not attained.
\end{example}

\begin{example}
Consider a distortion function of the form
\begin{equation}
\delta(x)=\frac{1}{\log(2)}\log(1+x) \ \ \ \mbox{for} \ x\in[0,1]\,.
\end{equation}
Then it is immediate to see that $\Indexfin(\rho_{\cA_\delta,S})=1$ and that the index is attained.
\end{example}


\bibliographystyle{plain}

\end{document}